\documentclass[submission]{eptcs}
\providecommand\event{CL\&C 2014}
\usepackage{ebutf8}

\title{A proof-theoretic view on scheduling in concurrency}
\author{Emmanuel Beffara
\institute{I2M, Université d'Aix-Marseille \& CNRS}}

\providecommand\titlerunning{A proof-theoretic view on scheduling in concurrency}
\providecommand\authorrunning{Emmanuel Beffara}

\usepackage{macros}

\usepackage{amsthm}
\newtheorem{definition}{Definition}
\newtheorem{example}{Example}
\newtheorem{theorem}{Theorem}
\newtheorem{proposition}{Proposition}
\newtheorem{lemma}{Lemma}

\begin{document}

\maketitle

\begin{abstract}
This paper elaborates on a new approach of the question of the proof-theoretic
study of concurrent interaction called ``proofs as schedules''. Observing that
proof theory is well suited to the description of confluent systems while
concurrency has non-determinism as a fundamental feature, we develop a
correspondence where proofs provide what is needed to make concurrent systems
confluent, namely scheduling. In our logical system, processes and schedulers
appear explicitly as proofs in different fragments of the proof language and
cut elimination between them does correspond to execution of a concurrent
system. This separation of roles suggests new insights for the denotational
semantics of processes and new methods for the translation of $π$-calculi into
prefix-less formalisms (like solos) as the operational counterpart of
translations between proof systems.
\end{abstract}

\section{Introduction}

The extension of the familiar Curry-Howard correspondence outside the
intuitionistic and functional worlds has been an active topic for decades,
with a variety of approaches to the question of determinism and confluence, or
lack thereof.
The interactive nature of cut elimination procedures suggested relationships
with actually interactive models of computation like games or process algebras.
Several systems were proposed based on linear logic~\cite{girard-1987-linear},
following the intuition that it is a logic of interaction.
Interpretations of proofs as processes, first sketched by
Abramsky~\cite{abramsky-1994-proofs} and formalized by Bellin and
Scott~\cite{bellin-1994--calculus}, later refined by various people including
the author of this paper~\cite{beffara-2006-concurrent-a}, stressed
that proof nets~\cite{girard-1996-proof-nets} and process calculi have
significant similarities in dynamics.
At the same time, type systems for concurrency~\cite{yoshida-2001-strong}
revealed to be equivalent to variants of linear
logic~\cite{honda-2010-exact,caires-2010-session}.
These approaches successfully stress the fact that concurrent calculi are
very expressive and versatile models of interactive behaviour, however they
are not satisfactory yet as a proof-theoretical account of concurrency, because they
tend to impose determinism in execution, effectively constraining processes to
essentially functional behaviour.

Getting out of this restriction requires to actually handle the inherent
non-determinism in concurrent systems.
A natural option is of course to relate processes with proofs in a
non-deterministic system.
The prototypical such system is the original classical sequent calculus LK,
however because of its structural rules it feels too remote from actual
process dynamics to offer a direct correspondence (although encodings through
intermediate systems have been proposed~\cite{vanbakel-2014-classical}).
Several more direct approaches have been proposed based on linear logic:
non-determinism in the style of complexity theory has been modelled using the
additives of linear
logic~\cite{maurel-2003-nondeterministic,mogbil-2010-non-deterministic,terui-2004-proof};
differential logic was recently developed by Ehrhard and
Regnier~\cite{ehrhard-2006-differential} and features a structured form of
non-determinism in its cut elimination; its untyped proof formalism was even
shown expressive enough to represent the
$π$-calculus~\cite{ehrhard-2010-interpreting}.

The present work elaborates on a different approach to the topic, first
sketched in a previous paper called \emph{Proofs as
executions}~\cite{beffara-2012-proofs}.
Our point of view is that cut elimination and general process execution should
not be made to match directly, because their meaning is on different levels.
On the one hand, the meaning of proofs lies in their normal forms, hence cut
elimination should be confluent in order to preserve meaning.
On the other hand, the meaning of a process is not its final irreducible form
but what happens to get there, as interaction with other processes (hence
execution of interactions should definitely not preserve meaning).
We thus establish a correspondence between proofs and interaction plans for
processes, hereafter called \emph{schedules} (rather than executions).
These schedules are what provides necessary information to make a system
deterministic, {i.e.} to make its execution confluent.
Note that despite the use of the word ``schedule'', we do not claim any precise
link with {e.g.} scheduling in operating systems; the word ``strategy'' might
also be appropriate, although the standard meaning of ``reduction strategy''
is way more restrictive than what we describe here.

Concretely, we develop this idea in a simple framework, relating schedules of
finite processes with proofs in multiplicative linear logic; this is to be
considered as a first step towards a wider correspondence, illustrating the
idea of our interpretation.
The process calculus we use is introduced in Section~\ref{sec:mccs}, the
logical system is introduced in Section~\ref{sec:mlla}.
Section~\ref{sec:exec} presents two logical translations of processes for which
the correspondence is proved, once in a synchronous form in
Section~\ref{sec:sync} and once in an asynchronous form in
Section~\ref{sec:async}, in which logic can abstract away from
execution order.
Section~\ref{sec:discussion} discusses extensions and directions for future
developments of the proofs-as-schedules paradigm, related to semantics of
processes (Section~\ref{sec:semantics}), the status of name hiding and passing
(Section~\ref{sec:hiding}) and the relationships with CPS translations and
determinisation in classical logic (Section~\ref{sec:cps}).

\medskip\emph{Note.}
The previous paper~\cite{beffara-2012-proofs} formalizes this idea by defining
a logical system, which appears as a kind of type system for processes, with
the crucial property that for every lock-avoiding execution of a process $P$
one can deduce a typing of $P$, hence a proof, whose cut elimination does
correspond to the considered execution.
A limitation of this result was that the ``type'' of a process could be very
different depending on its particular environment and execution, which made it
difficult to deduce an interpretation of processes in logical terms.
The present work improves the previous results in several respects:
Firstly, the typing of processes is uniform, and really independent of the
particular executions it might exhibit.
Subsequently, the contributions of the process and the scheduler in the proof
corresponding to a scheduling are clearly identified, as simple fragments of
the proof language.
This in turn suggests new approaches to the logical study of denotational
models of processes.
Despite these references to a previous work, this paper aims to be
self-contained.

\medskip\emph{Related work.}
The idea of matching proofs with executions is reminiscent of the proof-search
approach to computation.
Indeed, the relationship between logical linearity and interaction has been
explored (for instance by Miller and
Tiu~\cite{miller-1992--calculus,tiu-2005-proof}) but our approach has
different roots, in particular because of the status of cut-elimination,
nevertheless proof search in our context looks like the building of a schedule
for a fixed process, as illustrated in the technical proofs in this paper.
Bruscoli establishes a correspondence between proof search in a
variant of MLL and execution of a minimal CCS, in a context of
deep-inference~\cite{bruscoli-2002-purely}.
Although our translation of processes into formulas is different that that
used by the above works, connections can certainly be established, but we
defer this to later developments.

\section{Multiplicative CCS}
\label{sec:mccs}

We consider processes in a fragment of the standard language
CCS~\cite{milner-1989-communication}.
The fragment we use, hereafter called \emph{multiplicative} CCS (or MCCS), is
defined by the following grammar:
\begin{syntax}
  \define P, Q
  \case 1 \comment{inactive process}
  \case P\para Q \comment{parallel composition}
  \case \actp[ℓ]{a}P \comment{positive action prefix}
  \case \actn[ℓ]{a}P \comment{negative action prefix}
\end{syntax}
where $a$ is taken from a given set of \emph{channel names} and
$ℓ$ is taken from a given set $\Locs$ of \emph{locations}.
We impose that each location occurs at most once in any term: locations are
used to identify occurrences of actions in a process.
Note that we use $1$ for the inactive process instead of the usual $0$ because
it is the neutral element of $\para$ which is a multiplicative operation (in
the sense that it distributes over non-deterministic choice, which is then
additive, and not the other way around; besides it is in correspondence with
the unit $\mathbf{1}$ of linear logic and surely not with $\mathbf{0}$).

\begin{definition}
  Structural congruence is the smallest congruence $\equiv$ that makes
  parallel composition commutative and associative with $1$ as neutral
  element.
\end{definition}

\begin{definition}[execution]\label{def:execution}
  Execution is the relation over structural congruence classes, labelled
  by partial involutions over $\Locs$, defined by the rule
  \[
    \actn[ℓ]{a}P \para \actp[m]{a}Q \para R
    \quad \redpi{\{(\ell,m)\}} \quad
    P \para Q \para R
  \]
  Let $\redpix{}$ be the reflexive transitive closure of $\redpi{}$, with the
  annotations defined as $P\redpix{\emptyset}P$ and
  if $P\redpix{c}Q\redpix{d}R$ then $P\redpix{c\cup d}R$.
\end{definition}

Note that, by definition, $P\redpix{\emptyset}Q$ holds is when $P$ and $Q$ are
structurally congruent.
Moreover, the subscript $c$ in $P\redpix{c}Q$ is indeed an involution, {i.e.}
a set of disjoint pairs and not a list of pairs: it does not record the order
of interactions, only the information of which actions were synchronised
together; as stated in Theorem~\ref{thm:pairing}, this does characterize
executions up to permutation of independent transitions.

Locations do not affect the execution of processes: they are nothing more than
a technical tool used to name occurrences of actions in a formal way.
They will be used in Section~\ref{sec:sync} for the correspondence
between execution and cut elimination and in Section~\ref{sec:semantics} for
the discussion of the semantic interpretation of our results.
We introduce a few useful notations for this purpose:

\begin{definition}
  Let $P$ be an MCCS term.
  \begin{itemize}
  \item The set of locations occurring in $P$ is written $\Locs(P)$.
  \item Given $\ell\in\Locs(P)$, the \emph{subject} of $\ell$ is the name
    tagged by $\ell$, written $\subj[P]{\ell}$.
    The \emph{polarity} of $\ell$ is that of the action tagged by its subject,
    written $\pol[P]{\ell}$, element of $\{\pm1\}$.
  \item The action order of $P$ is the partial order $≤_P$ over $\Locs(P)$
    such that $\ell<_Pm$ for every sub-term $\actp[\ell]{x}Q$ of $P$ with
    $m\in\Locs(Q)$.
  \end{itemize}
\end{definition}

Remark that the information of locations, subjects, polarities and action
order completely characterizes a class of structural congruence of MCCS terms,
so we can consider that providing such information does define a process
without ambiguity.

When locations are unimportant for a given statement, they will be kept
implicit.
With this convention, the definition of the reduction relation is the standard
one:
\[
  \actp{a}P \para \actn{a}Q \para R
  \quad → \quad
  P \para Q \para R
\]

This language is very minimalistic and its operational theory is actually very
simple (one can prove that its bisimilarity is completely axiomatized by a
single rule scheme $(\actp{a}P)^{n+1}≃\actp{a}(P\para(\actp{a}P)^n)$ plus
structural congruence).
The results of this paper extend smoothly to a framework with replication and
sum.
A much more subtle point is that of name hiding, which is discussed in
Section~\ref{sec:hiding}.

\section{MLL with actions}
\label{sec:mlla}

Formulas of propositional multiplicative linear logic with
actions (in short MLLa) are generated from the following grammar:
\begin{syntax}
  \define A, B
  \case α \altcase α^⊥ \comment{propositional variable}
  \case A⊗B \altcase A⅋B \comment{conjunction, disjunction}
  \case \modp{a}{A} \altcase \modn{a}{A} \comment{action modalities}
\end{syntax}
where $α$ is taken from a given set of propositional variables and $a$ is
taken from the set of channel names.
Negation is defined inductively on formulas in the standard way, with
$(\modp{a}A)^⊥=\modn{a}(A^⊥)$;
linear implication $A⊸B$ is defined as a shorthand for $A^⊥⅋B$.

The modalities are reminiscent of those of
Hennessy-Milner~\cite{hennessy-1985-algebraic} logic, however the logic itself
is very different, in particular because of linearity.
While Hennessy-Milner modalities $\langle a\rangle A$ and $[a]A$ mean,
respectively, ``I can do $a$ and then satisfy $A$'' and ``whenever I do $a$, I
will satisfy $A$'', in our logic the formula $\modp{a}A$ means something like
``I will do $a$ and then exhibit behaviour $A$''.
There is no analogue of the modality $[a]$ because all the information we
provide is positive.
The duality induced by linear negation $(\cdot)^⊥$ is not a classical negation
but a change of roles: $A^⊥$ is the type of behaviours that interact correctly
with behaviours of type $A$.
The connectives $⊗$ and $⅋$ are not classical conjunction and disjunction
either, but spatial ones representing causality and independence between parts
of a run, using connectedness/acyclicity arguments to describe avoidance of
deadlocks.

\begin{table}
  \centering
\begin{tabular}{c@{\quad}c@{\hspace{4em}}c@{\quad}c}
  \begin{prooftree}
    \Infer0[ax]{ ⊢ A^⊥, A }
  \end{prooftree} &
  \begin{tikzpicture}[baseline={(0,0)}]
    \draw (0,0) to [out=90,in=90,looseness=1.5] node (ax) [midway,axiom] {} (2,0);
    \path (ax) node [above left] {$A^⊥$} node [above right] {$A$};
  \end{tikzpicture}
&
  \begin{prooftree}
    \Hypo{ ⊢ Γ, A }
    \Hypo{ ⊢ A^⊥, Δ }
    \Infer2[cut]{ ⊢ Γ, Δ }
  \end{prooftree} &
  \begin{tikzpicture}
    \node [subnet] (P) at (0,0) {$π$};
    \node [subnet] (Q) at (2,0) {$ρ$};
    \node (G) at (-0.5,-2) {};
    \node (D) at (2.5,-2) {};
    \draw (P) to [out=-80,in=-110,looseness=1.2] (Q);
    \draw (P) to [out=-120,in=90] (G) ;
    \draw (Q) to [out=-60,in=90] (D) ;
  \end{tikzpicture}
\\[3ex]
  \begin{prooftree}
    \Hypo{ ⊢ Γ, A }
    \Hypo{ ⊢ B, Δ }
    \Infer2[$⊗$]{ ⊢ Γ, A⊗B, Δ }
  \end{prooftree} &
  \begin{tikzpicture}
    \node [subnet] (P) at (0,0) {$π$};
    \node [subnet] (Q) at (2,0) {$ρ$};
    \node [link] (tens) at (1,-1.3) {$⊗$};
    \node (G) at (-0.5,-2.5) {};
    \node (c) at (1,-2.5) {};
    \node (D) at (2.5,-2.5) {};
    \draw (P) -- (tens);
    \draw (Q) -- (tens);
    \draw (tens) -- (c);
    \draw (P) to [out=-120,in=90] (G) ;
    \draw (Q) to [out=-60,in=90] (D) ;
  \end{tikzpicture}
&
  \begin{prooftree}
    \Hypo{ ⊢ Γ, A, B }
    \Infer1[$⅋$]{ ⊢ Γ, A⅋B }
  \end{prooftree} &
  \begin{tikzpicture}
    \node [subnet] (P) at (0,0) {$π$};
    \node [link] (par) at (0.7,-1.3) {$⅋$};
    \node (G) at (-0.8,-2.5) {};
    \node (c) at (0.8,-2.5) {};
    \draw (P) to [out=-80,in=120] (par);
    \draw (P) to [out=-40,in=60] (par);
    \draw (par) -- (c);
    \draw (P) to [out=-130,in=90] (G) ;
  \end{tikzpicture}
\\[5ex]
  \begin{prooftree}
    \Hypo{ ⊢ Γ, A }
    \Infer1[$\modp{a}$]{ ⊢ Γ, \modp{a}{A} }
  \end{prooftree} &
  \begin{tikzpicture}
    \node [subnet] (P) at (0,0) {$π$};
    \node [link] (mod) at (0.7,-1.4) {$\modp{a}$};
    \node (G) at (-0.6,-2.5) {};
    \node (c) at (0.6,-2.5) {};
    \draw (P) to [out=-50,in=90] (mod);
    \draw (mod) -- (c);
    \draw (P) to [out=-130,in=90] (G) ;
  \end{tikzpicture}
&
  \begin{prooftree}
    \Hypo{⊢ Γ, A }
    \Infer1[$\modn{a}$]{ ⊢ Γ, \modn{a}{A} }
  \end{prooftree} &
  \begin{tikzpicture}
    \node [subnet] (P) at (0,0) {$π$};
    \node [link] (mod) at (0.7,-1.4) {$\modn{a}$};
    \node (G) at (-0.6,-2.5) {};
    \node (c) at (0.6,-2.5) {};
    \draw (P) to [out=-50,in=90] (mod);
    \draw (mod) -- (c);
    \draw (P) to [out=-130,in=90] (G) ;
  \end{tikzpicture}
\end{tabular}
  \caption{Proof rules and proof net syntax for MLLa.}
  \label{proof-nets}
\end{table}

Proof rules are the standard rules of linear logic, extended with modalities,
as shown in Table~\ref{proof-nets}.
We use the language of proof nets: a \emph{proof structure} is a directed
graph whose nodes are labelled by names of proof rules and have the
appropriate input and output degrees for the rule they represent (premisses
are ingoing edges and conclusions are outgoing, the implicit orientation in
the pictures is downwards), a \emph{proof net} is a proof structure that is
the translation of a sequent calculus proof using the rules of
Table~\ref{proof-nets}.
As far as proof theory is concerned, these modalities are innocuous: they
commute with all other rules, their cancellation rule in cut elimination is
the obvious one, the types $A$ and $\modp{a}{A}$ are isomorphic.
Indeed, they are nothing more than markers in formulas (yet we call them
\emph{modalities} since they do have a logical meaning: they impose
restrictions on the structure of their possible proofs and this is the key for
the results in the present work).
For these reasons, standard theory for multiplicative proof nets applies to
proof nets of MLLa, including the Danos-Regnier correctness criterion:
\begin{theorem}[Danos-Regnier~\cite{danos-1989-structure}]
  Let $π$ be a proof structure.
  Define a \emph{switching graph} of $π$ as any graph obtained by deleting one
  of the ingoing edges of each $⅋$ node and forgetting about edge orientation.
  Then $π$ is a proof net if and only if all its switching graphs are
  connected and acyclic.
\end{theorem}

We restrict ourselves to propositional logic, however our construction heavily
relies on the instantiation of propositional variables with particular
formulas: this substitution mechanism is a key part of interaction as we model
it.
The same work could be carried out, possibly in a cleaner way, using
second-order quantification, however we decided not to use this quantification
explicitly because it makes proof theory more complicated and we do not need
its full power in the present work.
A similar approach was taken by Terui in his translation of boolean circuits
into MLL proofs~\cite{terui-2004-proof}.

In order to ease the formulation of operational correspondence proofs, we
assume that each proof structure comes with an injective labelling of its
modality links with locations in $\Locs$.
\begin{definition}
  Let $π$ be an MLLa proof structure (possibly with cuts).
  \begin{itemize}
  \item The set of locations occurring in $π$ is written $\Locs(π)$.
  \item Given $\ell\in\Locs(π)$, the \emph{subject} of $\ell$ is the name
    in the modality tagged by $\ell$, written $\subj[π]{\ell}$.
    The \emph{polarity} of $\ell$ is that of the modality, written
    $\pol[π]{\ell}$, element of $\{\pm1\}$.
  \item The proof order of $π$ is the partial order $≤_π$ over $\Locs(π)$
    such that $\ell<_πm$ when the link tagged $m$ appears in the tree of
    premisses of the link tagged by $ℓ$.
  \end{itemize}
\end{definition}

These definitions are similar to those for MCCS terms.
The partial order $≤_π$ is a bit more complicated to define but can be
explained in a simple way.
Observe that a proof structure $π$ can always be decomposed as a family of
trees of links together with a set of axioms between leafs and possibly a set
of cuts between roots:
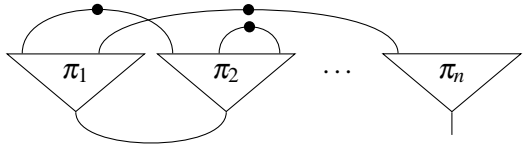
\begin{center}
  \begin{tikzpicture}[tree/.style={draw,isosceles triangle,
        isosceles triangle apex angle=100,
        shape border rotate=-90}]
    \node [tree] (p1) at (0,0) {$π_1$};
    \node [tree] (p2) at (2cm,0) {$π_2$};
    \node at (3.5cm,0) {$\cdots$};
    \node [tree] (pn) at (5cm,0) {$π_n$};
    \draw (p1.160) to [out=90,in=90] node [midway,axiom] {} (p2.160);
    \draw (p1.40) to [out=90,in=90,looseness=0.5] node [midway,axiom] {} (pn.160);
    \draw (p2.110) to [out=90,in=90,looseness=1.5] node [midway,axiom] {} (p2.20);
    \draw (p1.apex) to [out=-90,in=-90,looseness=0.7] (p2.apex);
    \draw (pn.apex) -- +(0,-3mm);
  \end{tikzpicture}
\end{center}
Then the partial order $≤_π$ over $\Locs(π)$ such that $ℓ≤_πm$ when the
modality link labelled $ℓ$ occurs below that labelled $m$ in one of the
trees $π_i$.

\section{Execution as implication}
\label{sec:exec}

In this section, we formally develop the correspondence between operational
semantics of processes and proofs in MLLa:
\begin{itemize}
\item Each process term $P$ is translated into a cut-free proof in MLLa, where
  instances of modality rules correspond to action prefixes in $P$; the
  conclusion $\ttype{P}$ of this proof (which can be considered as the
  type of $P$) is deduced syntactically from P.
\item A schedule for reducing $P$ to $Q$ is a \emph{multiplicative} proof
  (possibly with cuts) of the implication $\ttype{P}⊸\ttype{Q}$,
  {i.e.} this proof may not contain modality rules, which stresses the fact
  that it is only allowed to relate actions in a given process without
  introducing new actions.
\item An execution of $P$ corresponds to a cut elimination sequence of $P$ cut
  against a schedule for reducing $P$ to some $Q$; various cut elimination
  sequences correspond to different linear orderings of independent events in
  a given execution.
\end{itemize}
Hence the concurrent aspect of execution is represented by the variety of
proofs for a given implication, while different choices in cut elimination
correspond to unimportant ordering decisions, which is consistent with the
fact that cut elimination is confluent.

We implement this correspondence in two different ways.
The first version is \emph{synchronous} in that it exactly relates provability
of implication and cut elimination with step-by-step execution of processes.
The second version is \emph{asynchronous} in that it allows scheduling
decisions to be made in advance of execution, which leads to a more flexible
system, albeit with a more intricate interpretation.

\subsection{Synchronous translation}
\label{sec:sync}

In this translation, each term $P$ is mapped to a formula $\ttypes{P}$ that
has a unique cut-free proof, which follows the syntactic structure of $P$.
We then prove, in Theorem~\ref{thm:exec-sync}, that $\ttypes{P}⊸\ttypes{Q}$ is
provable if and only if there is an execution $P→^*Q$.

\begin{definition}[type assignment]\label{def:ttypes}
  Terms of MCCS are translated into MLLa formulas as follows, where in each
  case $α$ is a fresh propositional variable:
  \begin{align*}
    \ttypes{1} &:= α^⊥⅋α \\
    \ttypes{P\para Q} &:= \ttypes{P}⊗\ttypes{Q} \\
    \ttypes{\actp{a}P} &:= \modp{a}(α^⊥⅋(\ttypes{P}⊗α))
      &&= \modp{a}(α⊸(\ttypes{P}⊗α)) \\
    \ttypes{\actn{a}P} &:= \modn{a}(\ttypes{P}⊗α^⊥)⅋α
      &&= \modp{a}(\ttypes{P}⊸α)⊸α
  \end{align*}
\end{definition}

The freshness condition is a way to enforce polymorphism in our translation:
each proposition variable occurs exactly twice, once in each polarity.
As observed in Section~\ref{sec:mlla}, the intended meaning is indeed a
universal quantification over these variables, in a context where we chose not
to use second order quantification for simplicity.
For the same reason, the type for $1$ is the type of an identity, where we
could have used the multiplicative unit $\mathbf{1}$ of linear logic
(indeed, $\mathbf{1}$ and $∀α(α^⊥⅋α)$ are equivalent formulas): we
deliberately restrict ourselves to a unit-free logic in order to avoid the
slight proof-theoretic complications of the units.

\begin{proposition}[proof assignment]\label{prop:tproofs}
  For every MCCS term $P$, there is a unique cut-free proof of $\ttypes{P}$.
  This proof will be denoted as $\tproofs{P}$.
\end{proposition}
\begin{proof}
  The existence of a proof of $\ttypes{P}$ is proved by a straightforward
  induction on $P$.
  The case of action prefixes is as follows:
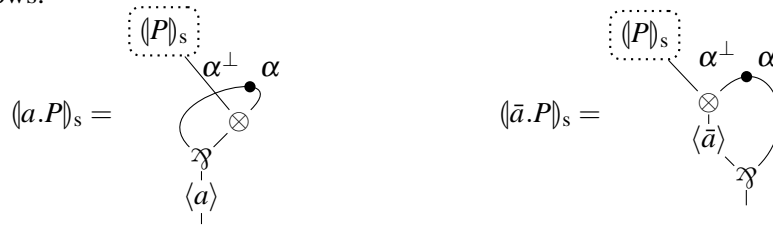
\begin{center}
  $\tproofs{\actp{a}P}={}$
  \begin{tikzpicture}
    \useasboundingbox (-2,0.5) rectangle (2,5.5);
    \node {}
        child {node [link] {$\modp{a}$}
          child {node (par) [link] {$⅋$}
            child [missing] {}
            child {node (tens) [link] {$⊗$}
              child [level distance=12mm,sibling distance=2cm]
                {node (P) [subnet] {$\tproofs{P}$}}
              child [missing] {}}}};
    \draw (tens) to[out=45,in=135,looseness=2.5]
      node [pos=0.3,axiom] (ax) {} (par);
    \draw (ax) node [above left] {$α^⊥$} node[above right] {$α$};
  \end{tikzpicture}
  \hfil
  $\tproofs{\actn{a}P}={}$
  \begin{tikzpicture}
    \useasboundingbox (-3.5,0) rectangle (1,5);
    \node {}
        child {node (par) [link] {$⅋$}
          child {node [link] {$\modn{a}$}
            child {node (tens) [link] {$⊗$}
              child [level distance=9mm,sibling distance=17mm]
                {node [subnet] {$\tproofs{P}$}}
              child [missing] {}}}
          child [missing] {}};
    \draw (tens) to[out=45,in=45,looseness=2]
      node [pos=0.2,axiom] (ax) {} (par);
    \draw (ax) node [above left] {$α^⊥$} node[above right] {$α$};
  \end{tikzpicture}
\end{center}
  The key point for uniqueness is that the freshness constraint on
  propositional variables imposes that for each variable $α$ there must be an
  axiom link between the occurrence of $α$ and that of $α^⊥$, and subsequently
  that all connectives must be explicitly introduced in a proof of
  $\ttypes{P}$.
\end{proof}

This proof assignment property is formulated in the absence of locations.
We enrich it in the obvious way with location information: the modality link
$\modp{a}$ that corresponds to a prefix $\actp[ℓ]{a}P$ gets labelled by the
location $ℓ$.

Theorem~\ref{thm:exec-sync} below states the correspondence between process
execution and provability of transitions.
In order to go from proofs to executions, we need to read back terms from
proofs, in order to get a kind of reverse operation for $\tproofs{-}$.
A direct read-back is hard to formulate, however we can define a particular
case of read-back for a proof that is derived from some $\tproofs{P}$ if we
know $P$.

\begin{definition}[term extraction]\label{def:tterms}
  Let $P$ be a term of MCCS.
  An MLLa proof structure $π$ is said to be \emph{compatible} with $P$ if
  $\Locs(π)⊆\Locs(P)$, subjects and polarities coincide between $P$ and $π$
  and the order $≤_P$ is included in the order $≤_π$.

  For $π$ compatible with $P$ we define the term $\restr{P}{π}$ to be the
  MCCS term with locations $\Locs(π)$, subjects and polarities as in $P$ and
  $π$ and action order as in $P$.
\end{definition}

Remark that for all terms $P$, by construction $\tproofs{P}$ is
compatible with $P$ and we have $\restr{P}{\tproofs{P}}\equiv P$.
Note also that compatibility implies that the process is \emph{less}
constrained than the proof: given a proof $π$, the set of processes compatible
with $π$ contains terms that are more parallel than the structure of $π$ (and
in particular the process that is just a parallel composition of all actions
that correspond to modality links in $π$).
The stricter ordering in the proof $π$ effectively means that $π$ has made
decisions on how $P$ will run, as illustrated by the lemma below and
Theorem~\ref{thm:exec-sync}.

\begin{lemma}\label{lemma:tterms}
  Assume $P$ is a term and $π$ be a proof structure compatible with $P$.
  For any cut-elimination step $π→π'$, the structure $π'$ is compatible with
  $P$ and one of the two following situations occurs:
  \begin{itemize}
  \item either $π→π'$ is an elimination of modality links $\modp{a}$ and
    $\modn{a}$ at some locations $ℓ$ and $m$, then
    $\restr{P}{π}\redpi{\{(ℓ,m)\}}\restr{P}{π'}$ is a valid execution step,
  \item or it is another elimination step and
    $\restr{P}{π}\equiv\restr{P}{π'}$.
  \end{itemize}
\end{lemma}
\begin{proof}
  The case of an elimination of modality links follows from the definition of
  term extraction: the actions in the extracted term are necessarily at
  top level since the modalities are premisses of a cut so their locations are
  minimal with respect to $≤_π$ hence with respect to $≤_P$ too.
  In $π'$ the location set is the same with $ℓ$ and $m$ removed and the other
  data (subjects, polarities and proof order) are simply restricted to this
  new set, so compatibility with $P$ is preserved.

  The other cut elimination steps can be either multiplicative eliminations or
  axiom eliminations.
  In these cases the location set is unchanged.
  For multiplicative steps the proof order is unchanged too, while in the case
  of axiom eliminations the proof order is possibly strengthened, since one
  tree of links in the proof is merged on top of another.
  In both cases, this preserves compatibility with $P$.
\end{proof}

Remark that the strengthening of the proof order in the case of axiom
elimination, {i.e.} the fact that if $π→π'$ by an axiom elimination then $≤_π$
may be strictly included in $≤_{π'}$, is the very reason why we need term
extraction: in this case the
syntactic structure of the proof is more constrained than that of the process
term $P$ we are observing, so we cannot read back a reduct of $P$ without
extra information.

\begin{theorem}\label{thm:exec-sync}
  Let $P$ and $Q$ be two MCCS terms.
  There is an execution $P→^*Q$ if and only if $\ttypes{P}⊸\ttypes{Q}$ is
  provable in MLL (without modality rules) for some instantiation of the
  propositional variables of $\ttypes{P}$.
\end{theorem}
\begin{proof}
  For the direct implication, remark that for a structural congruence
  $P\equiv Q$ at top level ({i.e.} not under prefixes), an implication
  $\ttypes{P}⊸\ttypes{Q}$ is provided by a standard isomorphism for the
  associativity and commutativity of the tensor, hence it is enough to handle the
  case of a reduction $(\actp{a}P\para\actn{a}Q)\para R→(P\para Q)\para R$.
  The expected type is
  \begin{multline*}
    \ttypes{(\actp{a}P\para\actn{a}Q)\para R} ⊸ \ttypes{(P\para Q)\para R} = \\
    ((\modn{a}(α⊗(\ttypes{P}^⊥⅋α^⊥)) ⅋ (\modp{a}(\ttypes{Q}^⊥⅋β)⊗β^⊥))
      ⅋ \ttypes{R}^⊥)
    ⅋ ((\ttypes{P} ⊗ \ttypes{Q}) ⊗ \ttypes{R})
  \end{multline*}
  so we get a proof as follows, if we instantiate $α$ as $\ttypes{Q}$ and $β$
  as $\ttypes{P}⊗\ttypes{Q}$:
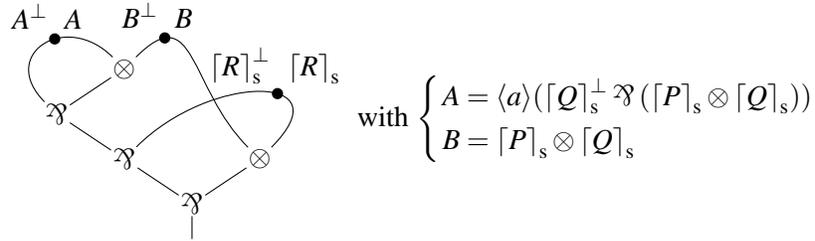
  \begin{center}
    \begin{tikzpicture}[sibling distance=18mm,level distance=6mm]
      \node {}
        child {node [link] {$⅋$}
          child {node [link] (parR) {$⅋$}
            child {node [link] (par2) {$⅋$}
              child [missing] {}
              child {node [link] (tens) {$⊗$}}}
            child [missing] {}}
          child {node [link] (tensR) {$⊗$}}};
      \draw (par2) to [out=135,in=135,looseness=2.1]
        node [pos=0.6,axiom] (aax) {} (tens);
      \draw (aax) node [above left] {$A^⊥\!$} node[above right] {$A$};
      \draw (tens) to [out=45,in=135,looseness=1.5]
        node (tax) [pos=0.2,axiom] {} (tensR);
      \draw (tax) node [above left] {$B^⊥\!$} node[above right] {$B$};
      \draw (parR) to [out=45,in=45,looseness=2]
        node (axR) [pos=0.6,axiom] {} (tensR);
      \draw (axR) node [above left] {$\ttypes{R}^⊥\!$} node[above right] {$\ttypes{R}$};
    \end{tikzpicture}
    with
    $\left\{
      \begin{aligned}
        A &= \modp{a}(\ttypes{Q}^⊥⅋(\ttypes{P}⊗\ttypes{Q})) \\
        B &= \ttypes{P}⊗\ttypes{Q}
      \end{aligned}
    \right.$
  \end{center}

  For the reverse implication, consider an MLL proof $π$ of
  $\ttypes{P}^⊥,\ttypes{Q}$ for some instantiation of the variables of
  $\ttypes{P}$.
  If we cut this proof against $\tproofs{P}$ (where the type variables are
  instantiated as in $π$), we get a proof $ρ$ of $\ttypes{Q}$.
  By construction, the cut-elimination procedure reduces each proof into a
  cut-free proof with the same conclusion, so cut-elimination in $ρ$
  reaches a cut-free proof of $\ttypes{Q}$; by Proposition~\ref{prop:tproofs}
  there is only one such proof, hence cut-elimination of $ρ$ reaches
  $\tproofs{Q}$.
  Since $π$ is an MLL proof, it contains no modality link, hence $ρ$ is
  compatible with $P$ according to Definition~\ref{def:tterms} and by
  Lemma~\ref{lemma:tterms} we know that each cut-elimination step of $ρ$
  induces either a structural congruence or an execution step in terms
  extracted from $P$ by intermediate proofs, so from a cut-elimination
  sequence $\tproofs{P}→^*\tproofs{Q}$ we can extract a CCS-execution sequence
  $P→^*Q$, and we also deduce that $\restr{P}{\tproofs{Q}}=Q$.
\end{proof}

\subsection{Asynchronous executions}
\label{sec:async}

The tight correspondence of Theorem~\ref{thm:exec-sync} is obtained thanks to
the fact that the action order $≤_P$ of processes is mapped to proof order
$≤_{\tproofs{P}}$ because of syntactic structure of the type $\ttypes{P}$, in
which modalities are nested as in $P$.
The downside of this correspondence is that it does not leave much space for
modular reasoning about the behaviour on different channels in processes.
Besides, it does not account for proofs extracted from executions as in our
previous work~\cite{beffara-2012-proofs}.

In this section, in order to generalise this result, we provide an asynchronous
variant of the correspondence.
We discuss below the relevance of this variant.

\begin{definition}[asynchronous type assignment]\label{def:ttypea}
  Terms of MCCS are translated into MLLa formulas as follows, where in each
  case $α$ is a propositional variable:
  \begin{align*}
    \ttypea{1} &:= α^⊥⅋α \\
    \ttypea{P\para Q} &:= \ttypea{P}⊗\ttypea{Q} \\
    \ttypea{\actp{a}P} &:= \modp{a}α^⊥⅋(\ttypea{P}⊗α)
      &&= \modn{a}α⊸(\ttypea{P}⊗α) \\
    \ttypea{\actn{a}P} &:= (\ttypea{P}⊗α^⊥)⅋\modn{a}α
      &&= (\ttypea{P}⊸α)⊸\modn{a}α
  \end{align*}
\end{definition}

Note that this translation is the same as the synchronous translation of
Definition~\ref{def:ttypes}, except for the position of modalities.

\begin{proposition}[asynchronous proof assignment]
  For every MCCS term $P$, there is a unique cut-free proof of $\ttypea{P}$.
  This proof will be denoted as $\tproofa{P}$.
\end{proposition}
\begin{proof}
  The argument is the same as in Proposition~\ref{prop:tproofs}.
  The case of action prefixes is now as follows:
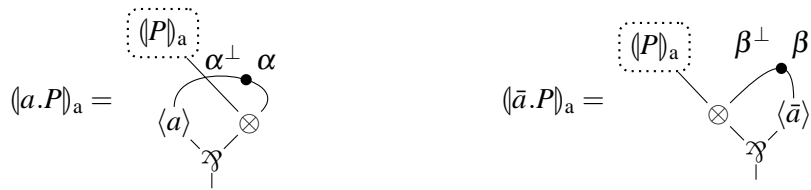
\begin{center}
  $\tproofa{\actp{a}P}={}$
  \begin{tikzpicture}
    \node {}
        child {node [link] {$⅋$}
          child {node (mod) [link] {$\modp{a}$}}
          child {node (tens) [link] {$⊗$}
            child [level distance=12mm,sibling distance=23mm]
              {node [subnet] {$\tproofa{P}$}}
            child [missing] {}}};
    \draw (mod) to[out=90,in=45,looseness=1.5]
      node [pos=0.6,axiom] (ax) {} (tens);
    \draw (ax) node [above left] {$α^⊥\!\!$} node[above right] {$α$};
  \end{tikzpicture}
  \hfil
  $\tproofa{\actn{a}P}={}$
  \begin{tikzpicture}
    \node {}
        child {node [link] {$⅋$}
          child {node (tens) [link] {$⊗$}
            child [level distance=9mm,sibling distance=17mm]
              {node [subnet] {$\tproofa{P}$}}
            child [missing] {}}
          child {node (mod) [link] {$\modn{a}$}
            child [missing] {}}};
    \draw (tens) to[out=45,in=90,looseness=2]
      node [pos=0.6,axiom] (ax) {} (mod);
    \draw (ax) node [above left] {$β^⊥$} node[above right] {$β$};
  \end{tikzpicture}
\end{center}
\end{proof}

\begin{proposition}\label{exec-async}
  For all MCCS reduction $P→Q$, the formula $\ttypea{P}⊸\ttypea{Q}$ is
  provable in MLL for some instantiation of the variables of $\ttypea{P}$.
\end{proposition}
\begin{proof}
  For any structural congruence $P\equiv Q$, the implication
  $\ttypea{P}→\ttypea{Q}$ is proved by a standard isomorphism for the
  monoidal structure of the tensor.
  For an execution step
  $(\actp{a}P\para\actn{a}Q)\para R→(P\para Q)\para R$,
  we prove the implication:
  \begin{multline*}
    \ttypea{(\actp{a}P\para\actn{a}Q)\para R} ⊸ \ttypea{(P\para Q)\para R} = \\
    ((\modn{a}α⊗(\ttypea{P}^⊥⅋α^⊥) ⅋ ((\ttypea{Q}^⊥⅋β)⊗\modp{a}β^⊥))
      ⅋ \ttypea{R}^⊥)
    ⅋ ((\ttypea{P} ⊗ \ttypea{Q}) ⊗ \ttypea{R})
  \end{multline*}
  with the following proof:
  
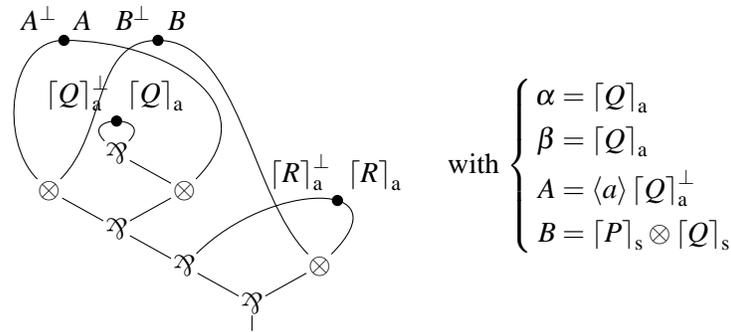
\begin{center}
    \begin{tikzpicture}[sibling distance=18mm,level distance=5mm]
      \node {}
        child {node [link] {$⅋$}
          child {node [link] (parR) {$⅋$}
            child {node [link] (parPQ) {$⅋$}
              child {node [link] (tensP) {$⊗$}}
              child {node [link] (tensQ) {$⊗$}
                child {node [link] (parQ) {$⅋$}}
                child [missing] {}}}
            child [missing] {}}
          child {node [link] (tensR) {$⊗$}}};
      \draw (tensP) to [out=135,in=180] (-5,8) to [out=0,in=45,looseness=1.4]
        node [pos=0,axiom] (axA) {} (tensQ);
      \draw (axA) node [above left] {$A^⊥\!$} node[above right] {$A$};
      \draw (tensP) to [out=45,in=180] (-2.5,8) to [out=0,in=135,looseness=0.8]
        node [pos=0,axiom] (axB) {} (tensR);
      \draw (axB) node [above left] {$B^⊥\!$} node[above right] {$B$};
      \draw (parQ) to [out=135,in=45,looseness=4]
        node [pos=0.5,axiom] (axQ) {} (parQ);
      \draw (axQ) node [above left] {$\ttypea{Q}^⊥\!$} node[above right] {$\ttypea{Q}$};
      \draw (parR) to [out=45,in=45,looseness=2]
        node [pos=0.6,axiom] (axR) {} (tensR);
      \draw (axR) node [above left] {$\ttypea{R}^⊥\!$} node[above right] {$\ttypea{R}$};
    \end{tikzpicture}
    \quad with
    $\left\{
      \begin{aligned}
        α &= \ttypea{Q} \\
        β &= \ttypea{Q} \\
        A &= \modp{a}\ttypea{Q}^⊥ \\
        B &= \ttypes{P}⊗\ttypes{Q}
      \end{aligned}
    \right.$
  \end{center}
\end{proof}

\begin{lemma}\label{step-async}
  Let $P$ be an MCCS term with at least one action prefix.
  If $\ttypea{P}⊸\ttypea{1}$ is provable in MLL for some instantiation of the
  propositional variables, then there exists a reduction $P→Q$ such that
  $\ttypea{Q}⊸\ttypea{1}$ is also provable.
\end{lemma}
Note that the condition on $P$ is that it is not already congruent to $1$,
otherwise $\ttypea{P}⊸\ttypea{1}$ is provable in MLL (as well as the reverse
implication) but obviously $P$ does not reduce.
However, the statement does \emph{not} require that $P$ must reducible ({i.e.}
that it is not in “normal form” for execution): the existence of a proof of
$\ttypea{P}⊸\ttypea{1}$ does imply this fact.
\begin{proof}
  Consider a proof $π$ of $\ttypea{P}^⊥⅋\ttypea{1}$ for some instantiation of
  the propositional variables of $\ttypea{P}$.
  MLL admits cut elimination and expansion of axiom links, so we may assume
  without loss of generality that $π$ is cut-free and that all multiplicative
  connectives in $\ttypea{P}^⊥$ are introduced by $π$.

  The term $P$ can be written as a composition of action prefixes
  $\actp{a_1}P_1\para\cdots\para\actp{a_n}P_n$ (where the $a_i$ are actions of
  arbitrary polarities) so $π$ is decomposed as follows, with some
  tensors possibly reversed depending on the polarities of the $a_i$:
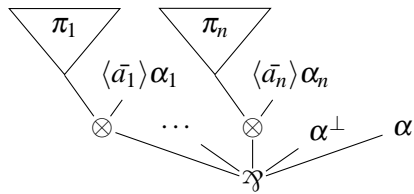
  \begin{center}
    \begin{tikzpicture}[level distance=7mm]
      \tikzset{subtree/.style={
          draw,shape=isosceles triangle,anchor=apex,
          edge from parent/.style={draw,child anchor=apex},
          isosceles triangle apex angle=80,
          shape border rotate=-90,}}
      \node [link] {$⅋$}
        child {node [link] {$⊗$}
          child [subtree] {node [subtree] {$π_1$}}
          child {node {$\modn{a_1}α_1$}}}
        child {node {$\cdots$}}
        child {node [link] {$⊗$}
          child [subtree] {node [subtree] {$π_n$}}
          child {node {$\modn{a_n}α_n$}}}
        child {node {$α^⊥$}}
        child {node {$α$}};
    \end{tikzpicture}
  \end{center}
  with axiom links on top (the $⅋$ link of arity $n+2$ in the picture
  actually stands for some tree of binary $⅋$ links, corresponding to the
  bracketing of parallel compositions in $P$, since our language only includes
  binary connectives).
  Note that, since $π$ does not use modality rules, all premisses
  $\modn{a_i}α_i$ must be introduced by axiom links.

  We now argue that two of these are actually connected by an axiom rule.
  Towards a contradiction, assume that it is not the case: every
  $\modn{a_i}α_i$ is introduced by an axiom link whose other conclusion is a
  leaf in one of the subproofs $π_{f(i)}$.
  Since there are finitely many subproofs, the sequence $1,f(1),f^2(1),…$ is
  eventually periodic.
  Consider a minimal cycle $i,f(i),f^2(i),…,f^k(i)=i$.
  Up to the reordering of sub-terms of $P$, let us assume this sequence is
  $1,2,…,k,1$.
  Then we have the following situation:
  
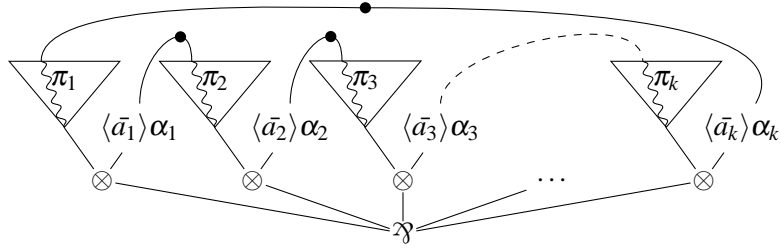
\begin{center}
    \begin{tikzpicture}[level distance=7mm]
      \node [link] {$⅋$}
        child {node [link] {$⊗$}
          child [subtree] {node [subtree] (p1) {$π_1$}}
          child {node (a1) {$\modn{a_1}α_1$}}}
        child [missing] {}
        child {node [link] {$⊗$}
          child [subtree] {node [subtree] (p2) {$π_2$}}
          child {node (a2) {$\modn{a_2}α_2$}}}
        child [missing] {}
        child {node [link] {$⊗$}
          child [subtree] {node [subtree] (p3) {$π_3$}}
          child {node (a3) {$\modn{a_3}α_3$}}}
        child [missing] {}
        child {node {$\cdots$}}
        child [missing] {}
        child {node [link] {$⊗$}
          child [subtree] {node [subtree] (pk) {$π_k$}}
          child {node (ak) {$\modn{a_k}α_k$}}};
      \draw (a1) to [out=90,in=90,looseness=2] node[pos=0.7,axiom] {} (p2.140)
        decorate [subtree path] { -- (p2.apex)};
      \draw (a2) to [out=90,in=90,looseness=2] node[pos=0.7,axiom] {} (p3.140)
        decorate [subtree path] { -- (p3.apex)};
      \draw[dashed] (a3) to [out=90,in=90,looseness=0.7] (pk.140);
      \draw decorate [subtree path] {(pk.140) -- (pk.apex)};
      \draw (ak) to [out=60,in=0,looseness=0.7] (-1,6) node [axiom] {}
        to [out=180,in=90,looseness=0.5] (p1.140)
        decorate [subtree path] { -- (p1.apex)};
    \end{tikzpicture}
  \end{center}
  Call $t_i$ the tensor link under the node $\modn{a_i}α_i$, and $x_i$ the
  axiom link that has $\modn{a_i}α_i$ as one of its premisses.
  Then we have a cycle that goes from $t_1$ up to $x_1$ into $π_2$, down to
  $t_2$, up to $x_2$ into $π_3$ and so on until $t_k$, up to $x_k$ and back
  into $π_1$, down to $t_1$.
  This cycle traverses each tree $π_i$ straight from a leaf to the root, so it
  only goes through at most one premiss of each $⅋$ link.
  Hence there is a $⅋$-switching of $π$ that has a cycle, which violates the
  Danos-Regnier correctness criterion for multiplicative proof
  nets~\cite{danos-1989-structure}.

  Hence there is an axiom link between two of the $\modn{a_i}α_i$.
  For simplicity, assume that the indices are $1$ and $2$.
  Then $a_1$ and $a_2$ are dual, so $P$ can be written
  $\actn{a}P_1\para\actp{a}P_2\para P'$, and we have
  $P→P_1\para P_2\para P'$.
  Moreover the cut between $\tproofa{P}$ and $π$ has the following shape:
  
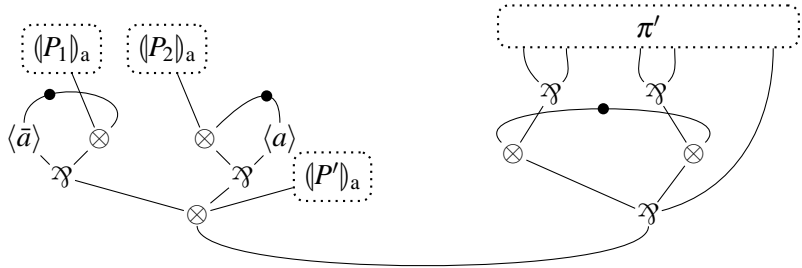
\begin{center}
    \begin{tikzpicture}[
        level 1/.style={sibling distance=12mm},
        level 2/.style={sibling distance=10mm}]
      \draw (-12,0) node [link] (r1) {$⊗$}
        child {node [link] {$⅋$}
          child {node [link] (ap1) {$\modn{a}$}}
          child {node [link] (tp1) {$⊗$}
            child [level distance=12mm] {node [subnet] {$\tproofa{P_1}$}}
            child [missing] {}}}
        child [missing] {}
        child {node [link] {$⅋$}
          child {node [link] (tp2) {$⊗$}
            child [level distance=12mm] {node [subnet] {$\tproofa{P_2}$}}
            child [missing] {}}
          child {node [link] (ap2) {$\modp{a}$}}}
        child {node [subnet] {$\tproofa{P'}$}};
      \draw (ap1) to [out=90,in=45,looseness=1.5] node [pos=0.3,axiom] {} (tp1);
      \draw (tp2) to [out=45,in=90,looseness=1.8] node [pos=0.6,axiom] {} (ap2);
      \node [subnet,minimum width=4cm] at (0,5) (pi) {$π'$};
      \draw [level distance=8mm,level 1/.style={sibling distance=12mm}]
       node [link] (r2) {$⅋$}
        child {node [link] (t1) {$⊗$}
          child [missing] {}
          child {node [link] (p1) {$⅋$}}}
        child [missing] {}
        child {node [link] (t2) {$⊗$}
          child {node [link] (p2) {$⅋$}}
          child [missing] {}}
        child [missing] {};
      \draw (t1) to [out=135,in=45,looseness=0.8]
        node [pos=0.5,axiom] (ax) {} (t2);
      \draw (r1) to [out=-90,in=-90,looseness=0.3] (r2);
      \draw (p1) to [out=135,in=-90] (pi.190);
      \draw (p1) to [out=45,in=-90] (pi.195);
      \draw (p2) to [out=135,in=-90] (pi.250);
      \draw (p2) to [out=45,in=-90] (pi.320);
      \draw (r2) to [out=15,in=-90] (pi.350);
    \end{tikzpicture}
  \end{center}
  After a few steps of cut elimination (five pairs of multiplicatives, three
  axioms and a pair of modalities) this proof eventually reduces into $π'$ cut
  against $\tproofa{P_1}$, $\tproofa{P_2}$ and $\tproofa{P'}$ plus a cut
  between the two remaining ports of $π'$.
  This reduct is also a reduct of $\tproofa{P_1\para P_2\para P'}$ against
  $π'$ plus two $⅋$ rules.
  The latter is thus an MLL proof of
  $\ttypea{P_1\para P_2\para P'}⊸\ttypea{1}$.
\end{proof}

\begin{theorem}\label{thm:exec-async}
  Let $P$ be an MCCS term.
  There is an execution $P→^*1$ if and only if $\ttypea{P}⊸\ttypea{1}$ is
  provable in MLL (without modality rules) for some instantiation of the
  propositional variables.
\end{theorem}
\begin{proof}
  For the direct implication, Proposition~\ref{exec-async} provides the
  implication for one step of execution.
  The result follows immediately since linear implications properly compose
  and the correctness of proofs is preserved by instantiation of propositional
  variables.

  For the reverse implication, we can reason by induction on the number of
  action prefixes in $P$.
  The base case of a term with no action is vacuous since the neutral process
  $1$ is the only such term, up to structural congruence.
  Lemma~\ref{step-async} provides the induction step.
\end{proof}

The interpretation of Theorem~\ref{thm:exec-async} is more subtle than for
Theorem~\ref{thm:exec-sync}, because in this variant there is no step-by-step
correspondence between execution and cut elimination.
Of course, Lemma~\ref{step-async} does prove that there is a \emph{strategy}
for eliminating cuts that corresponds to an execution, but in general the
normalisation of $\tproofa{P}$ against a particular proof of
$\ttypea{P}⊸\ttypea{1}$ may consume modalities in $P$ in an arbitrary order.
This is why Theorem~\ref{thm:exec-async} requires the final type to be
$\ttypea{1}$: the statement like of Theorem~\ref{thm:exec-sync} is actually
false for the asynchronous variant.

\begin{definition}
  Let $⇒$ be the relation over processes such that $P⇒Q$ holds whenever
  $\ttypea{P}⊸\ttypea{Q}$ is provable in MLL for some instantiation of the
  propositional variables of $\ttypea{P}$.
\end{definition}

By Proposition~\ref{exec-async} we know that $P→^*Q$ implies $P⇒Q$, but other
rules are easily provable, like the following:
\begin{center}
  \begin{prooftree}
    \Hypo{ P &→ Q }
    \Infer1{ P &⇒ Q }
  \end{prooftree}
\hfil
  \begin{prooftree}
    \Hypo{ P &⇒ Q }
    \Infer1{ a.P &⇒ a.Q }
  \end{prooftree}
\hfil
  $\actp{u}\actp{a}P\para\actp{v}\actn{a}Q ⇒ \actp{u}P\para\actp{v}Q$
\end{center}
Therefore $⇒$ can make decisions in advance about the execution of a process,
and update the visible part of the process accordingly.
In other words, $⇒$ is an fully asynchronous form of execution: it allows to
perform an interaction as soon as this interaction may be part of some full
execution.

\section{Discussion}
\label{sec:discussion}

We thus have defined two type assignment systems for MCCS processes into MLLa
that enjoy good properties relating cut-elimination and execution.
It appears that processes correspond to cut-free proofs using modality links,
while schedules correspond to pure MLL proofs without modality links.

Although the technical details of the paper are limited to multiplicative CCS,
it is clear that the same approach extends to the full calculus, including
replication and sum, with a type assignment like
\begin{align*}
  \ttypea{\oc P} &= \oc\ttypea{P} &
  \ttypea{P+Q} &= \ttypea{P}\with\ttypea{Q}
\end{align*}
so that, for replication, a schedule has to decide the number of copies of $P$
it will use, and for the sum it will have to decide which side of the additive
conjunction will be effectively used.
This leads to a correspondence using full propositional linear logic; the
theory of proof nets is harder in this case but the principles of the
correspondence remain the same.

We now conclude with ideas for research directions open by the present
work.

\subsection{Semantics of processes}
\label{sec:semantics}

The main source of non-determinism is the fact that a given
action name may occur several times in a given term, and locations are used to
name the different occurrences.

The annotation $c$ in an execution step $P\redpi{c}Q$ describes which
occurrences interact.
Remark that, for a given $P$ and $c$, there is at most one $Q$
such that $P\redpi{c}Q$, since $c$ describes the interaction completely.

\begin{definition}[pairing]\label{def:pairing}
  A \emph{pairing} of a term $P$ is a partial involution $c$ over $\Locs(P)$
  such that for all $\ell\in\domain c$,
  $\subj{c(\ell)}=\subj{\ell}$ and
  $\pol{c(\ell)}=-\pol{\ell}$.

  Let $\sim_c$ be the smallest equivalence that contains
  $c$.
  $c$ is \emph{consistent} if $\domain{c}$ is downward closed for
  $≤_P$ and $\sim_c<_P\sim_c$ is acyclic.
\end{definition}

\begin{example}\label{ex:pairing}
The total pairings of $P = a^1.c^2 \para b^3.\bar{a}^4 \para
\bar{b}^5.\bar{c}^6 \para a^7.\bar{b}^8 \para b^9 \para \bar{a}^0 $ are \\
$c_1=\{ (9,5), (1,0), (2,6), (3,8), (4,7) \}$,
$c_2=\{ (3,5), (1,4), (2,6), (7,0), (9,8) \}$,\\
$c_3=\{ (1,4), (3,8), (7,0), (9,5), (2,6) \}$,
$c_4=\{ (1,0), (3,5), (7,4), (9,8), (2,6) \}$.\\
Only $c_1$ is inconsistent as there is a cycle induced by $\{(3,8), (4,7) \}$.
The maximal consistent pairing included in $c_1$ is $\{(9,5),(1,0),(2,6)\}$.
\end{example}

Observe that pairings and consistency are preserved by structural congruence,
as a direct consequence of the fact that subjects, polarities and
prefixing are preserved by structural congruence.
In our previous work~\cite{beffara-2012-proofs}, we established the following
precise relationship between pairings and executions:

\begin{theorem}\label{thm:pairing}
  Let $P$ be an MCCS term and $c$ a pairing of $P$.
  \begin{itemize}
  \item $c$ is consistent if and only if there is a term $Q$ such that
    $P\redpix{c}Q$,
  \item any two executions $P\redpix{c}Q$ and $P\redpix{c}R$ with the same
    pairing are permutations of each other, and in this case $Q\equiv R$.
  \end{itemize}
\end{theorem}

Maximal consistent pairings represent executions of processes until a state
where no more execution is possible.

We can actually relate pairings and logical schedules in a precise way.
Consider a term $P$ and proof $π$ of $\ttypea{P}⊸\ttypea{1}$.
In a cut elimination process of $\tproofa{P}$ against $π$, eventually each
modality link in $\tproofa{P}$ will be eliminated by a dual link also from
$\tproofa{P}$, and this association between modality links is of course
independent of the cut elimination sequence since proof normalisation is
confluent.
Hence $π$ induces a pairing between modality links in $\tproofa{P}$, and since
modality links are in bijection with locations in $P$, this in turn induces a
pairing of $P$.
This pairing is maximal since it reaches all locations in $P$, and by
Theorem~\ref{thm:exec-async} it is consistent. Reciprocally, every maximal
consistent pairing of $P$ induces a class of proofs of
$\ttypea{P}⊸\ttypea{1}$.

Moreover, this relationship between pairings and proofs is actually a
correspondence between pairings of $P$ and possible ways of positioning axiom
links on top of a canonical multiplicative proof structure for
$\ttypea{P}⊸\ttypea{1}$, as observed in the proof of Lemma~\ref{step-async}.
However, this correspondence is not straightforward (like one axiom for each
pair) because there are more leaves in proofs of this type than there are
actions in $P$, moreover the instantiation of propositional variables is far
from obvious since, as in the proof of Proposition~\ref{exec-async}, it may
include the whole type of a term involved in the execution.
We defer the precise study of this relationship to future work.

\bigskip

The next step in semantics is the study of denotational models of proofs in
our system, which should provide new insights for the denotational semantics
of processes.
For instance, it is a trivial remark that coherence
spaces~\cite{girard-1987-linear} have a formal similarity with event
structures~\cite{winskel-1987-event}, although a precise relationship is hard
to formulate, especially in a proofs-as-processes approach.
Since our constructions maps processes to types, {i.e.} spaces, and schedules
to proofs {i.e.} cliques or configurations in spaces, it should be possible to
formalize such a relationship in our setting in a meaningful way.
A difficulty in this respect is of course to define a non-trivial
interpretation of modalities: this is needed in order to identify cliques that
are interpretations of multiplicative proofs, as opposed to general proofs
with modalities.

\subsection{Name hiding and passing}
\label{sec:hiding}

The question of how to integrate name hiding in our system is crucial as it is
necessary for expressiveness and modularity.
However, Theorems \ref{thm:exec-sync}~and~\ref{thm:exec-async}, which
characterize execution in the proof system, do not extend to a calculus with
name hiding.
Indeed, a reasonable type assignment for a term $\new{a}P$ should not let $a$
appear in the formula $\lceil\new{a}P\rceil$, yet actions on this bound $a$
should be scheduled.

A syntactic way of handling name hiding could be to allow quantification over
channel names in the logic, so that $\new{a}$ in a process would become a
quantification in the type.
A schedule would then have to handle actions on this name without knowing
anything about it, hence considering it effectively as a fresh name.
However, the nature of such a quantifier is unclear:
\begin{itemize}
\item A universal $∀a$ is not an option, since it would allow the scheduler to
  provide any name for $a$ in $\new{a}P$, including one that $P$ or its
  environment already knows, which would produce new possible interactions.
\item An existential $∃a$ is better suited since it forces the scheduler to
  behave independently of the actual name used, effectively treating $a$ as
  fresh.
  However, a proper correspondence can only be obtained only if, by
  construction, $∃a$ is only introduced with a premiss where $a$ is fresh, and
  such a restriction is very unnatural as it does not respect the meaning of
  connectives in the target logic.
\item Using a specific quantifier for freshness, like Miller and Tiu's
  $\nabla$~\cite{miller-2005-proof}, would probably solve this issue but at
  the price of introducing an ad-hoc construct with its particular theory in
  an otherwise well-studied logic, with the consequence that it would make the
  semantic explorations mentioned above more difficult.
\end{itemize}

A very different approach, more semantic than the use of a quantifier, would
be to decide that placing a $\new{a}$ in front of a process would mean
something like making decisions about the scheduling that will happen on name
$a$, by cutting the relevant conclusions against a partial schedule ({i.e.} a
proof without modality rules).
\begin{itemize}
\item Only the asynchronous translation of section~\ref{sec:async} can handle
  this approach, since Theorem~\ref{thm:exec-sync} proves that any partial
  schedule in the synchronous variant is the beginning of an execution.
  This is consistent with the fact that it amounts to scheduling a channel in
  advance of actual execution.
\item It blurs the distinction between a process and a scheduler, since now a
  process does include scheduling decisions.
  This issue could be solved by typing a process $\new{a}P$ with something
  like $S⊸\ttypea{P'}$ where $S$ is the type of a scheduler for a channel
  (depending on how $a$ is used in $P$) and $P'$ is a type for $P$ where $a$
  is hidden.
\end{itemize}

The latter approach would most likely require second-order quantification in
order to specify $S$.
Further study of this point should also be the way to approach the problem of
extending our work to name passing calculi: once we know how to properly
specify what a channel is for restricting it, we can investigate how to
communicate it through another channel.

\medskip

\subsection{CPS translations and determinisation}
\label{sec:cps}

Both our type assignments actually translate actions and processes into types
of \emph{intuitionistic} MLL formulas, therefore proofs extracted from
processes could be interpreted as linear functions.
Moreover, the structure of the translations, up to the position of modalities,
is $\lfloor\actn{a}P\rfloor=∀α.(\lfloor P\rfloor⊸α)⊸α$, which is a weak form of
double-negation, so weak in fact that this formula, without modalities, is
isomorphic to $\lfloor P\rfloor$.
This suggests a kind of (linear) continuation-passing-style translation of
processes into a linear $λ$-calculus, which could be another setting for the
semantic and operational study of processes and schedules.
In relation with other forms of proofs-as-processes correspondences, this
might suggest an approach to the question of desequentialisation of processes
(like translations into
solos~\cite{beffara-2006-concurrent,laneve-2003-solos}) as a translation
between logical systems, in a way similar to CPS translations from classical
logic into intuitionistic logic.

This idea of CPS-like translations indeed suggest a step in a wider programme
for the approach of non-determinism in logic.
It is well known that the classical sequent calculus LK is non confluent, to
the point that any two proofs of the same formula can be identified by
conversion through cut-elimination, which makes any direct denotational
semantics of proofs degenerate, and that interesting semantics can be obtained
by means of double-negation translations into intuitionistic logic.
Such translations are effective because they impose constraints on proof
reduction, which on the computational side corresponds to fixing reduction
strategies like call-by-name or call-by-value in languages with control.
It is a form of determinisation, and such translations have been extensively
studied.
The systematic study of decompositions of LK into linear logic by Danos, Joinet and
Schellinx in the system LKtq~\cite{danos-1997-new} is of particular interest
with respect to the present work since it combines determinisation and
linearisation.
Reading these translations from the computational point of view provides
translations of classical proofs into a deterministic process calculus.

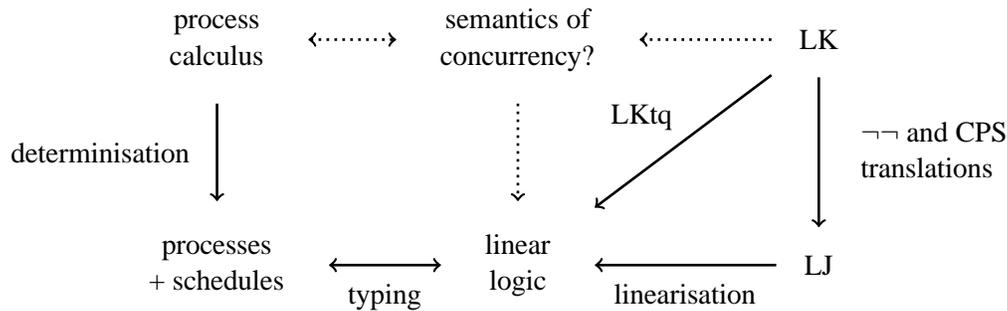
\begin{center}
  \begin{tikzpicture}[
      point/.style={shape=rectangle,inner sep=1ex,outer sep=2mm},
      link/.style={line width=1pt}]
    \node[point] (pc) at (0,3cm)
      {\begin{tabular}{c}process\\calculus\end{tabular}};
    \node[point] (sched) at (0,0)
      {\begin{tabular}{c}processes\\+ schedules\end{tabular}};
    \node[point] (ll) at (4cm,0)
      {\begin{tabular}{c}linear\\logic\end{tabular}};
    \draw[link,->] (pc) to node[pos=0.5,left=2mm] {determinisation} (sched);
    \draw[link,<->] (sched) to node[pos=0.5,below=1mm] {typing} (ll);
    \node[point] (lk) at (8cm,3cm) {LK};
    \node[point] (lj) at (8cm,0) {LJ};
    \draw[link,->] (lk) to node[pos=0.5,right=2mm]
      {\begin{tabular}{l}$¬¬$ and CPS\\translations\end{tabular}} (lj);
    \node[point] (x) at (4cm,3cm)
      {\begin{tabular}{c}semantics of\\concurrency?\end{tabular}};
    \draw[link,->] (lj) to node[pos=0.5,below=1mm] {linearisation} (ll);
    \draw[link,->] (lk) to node[pos=0.5,above left] {LKtq} (ll);
    \draw[link,dotted,->] (x) -- (ll);
    \draw[link,dotted,->] (lk) -- (x);
    \draw[link,dotted,<->] (pc) -- (x);
  \end{tikzpicture}
\end{center}

Our approach of non-determinism by means of scheduling suggest a different
take on these questions: would it be possible to provide a form of
linearisation of classical proofs that would preserve the intrinsic
non-determinism of LK, so that various evaluations strategies could be
obtained by effectively choosing different types of schedulers?
The semantic study of our proofs-as-schedules correspondence could be an
approach to this kind of question by means of a study of the semantics of
processes.

\providecommand\urlalt[2]{\href{#1}{\url{#2}}}
\providecommand\urlprefix[3]{}
\bibliography{zotero}
\bibliographystyle{eptcs}

\end{document}